\documentclass{article}
\usepackage{amsmath,amssymb,amsthm}

\newtheorem{Theorem}{Theorem}[section]

\newtheorem{Lemma}[Theorem]{Lemma}
\newtheorem{Proposition}[Theorem]{Proposition}
{\theoremstyle{definition}

\newtheorem{Example}[Theorem]{Example}

\newtheorem{Remark}[Theorem]{Remark}
}

\usepackage{url,graphicx}
\usepackage{caption}
\usepackage{subcaption}
\usepackage{todonotes}

\renewcommand{\P}{\mathbb P}
\newcommand{\C}{\mathbb C}
\newcommand{\R}{\mathbb R}
\renewcommand{\H}{\mathbb H}
\newcommand{\D}{\mathbb D}

\newcommand{\qi}{\mathbf i}
\newcommand{\qj}{\mathbf j}
\newcommand{\qk}{\mathbf k}
\newcommand{\ci}{\mathrm i}
\newcommand{\eps}{\epsilon}
\newcommand{\SE}{\mathrm{SE3}}

\title{A Survey on the Theory of Bonds}

\author{Zijia Li\thanks{This research was supported by the Austrian Science Fund (FWF): P26607-N25, and by
	the Austrian Ministry for Transport, Innovation and Technology (BMVIT) within the framework of the sponsorship
agreement formed for 2015-2018 under the project RedRobCo.} , Joanneum Research, Klagenfurt, Austria \\
	Josef Schicho$\thanks{This research was supported by the Austrian Science Fund (FWF): P26607-N25.}$, RISC, University of Linz, Austria \\
	Hans-Peter Schr\"ocker$^\dagger$, University of Innsbruck, Austria}

\begin{document}

\maketitle

\begin{abstract}
  Many researchers tried to understand/explain the geometric reasons for paradoxical mobility of
  a mechanical linkage, i.e. the situation when a linkage allows more motions than expected
  from counting parameters and constraints. Bond theory is a method that aims at understanding
  paradoxical mobility from an algebraic point of view. Here we give a self-contained
  introduction of this theory and discuss its results on closed linkages with revolute or
  prismatic joints.
\end{abstract}

\section*{Introduction}

By definition, the mobility of a mechanical linkage is the dimension of its configuration space.
We say that a mechanical linkage moves paradoxically if the mobility is positive, but
one does not expect this by counting parameters and constraints. We are especially interested
in the case when the expected mobility is zero, but the linkage is still mobile.
Examples are closed linkages
with 6 revolute joints: for a generic choice of parameters, the closure equations have 16
complex solutions. But there are many families of special cases with mobility~1, such as
Hooke's linkage~\cite{Hooke}, Bricard's line symmetric linkage\cite{bricard}, 
or Wohlhart's partially symmetric linkage~\cite{partsym}.

The theory of bonds was introduced in \cite{hss2} as a tool for systematically explaining
and analyzing paradoxical mobility of closed loops with only revolute joints. That paper
contains a simplified proof of Karger's classification of mobile closed 5R linkages (Karger's
original proof \cite{karger} uses computer algebra). In \cite{hlss}, the theory is used
to prove that the genus of the configuration curve of a mobile 6R linkage is at most 5, and
to classify all cases where the maximum is attained. In \cite{ls}, the theory is used to
obtain equations in the Denavit/Hartenberg parameters of a 6R linkage that are necessary
for mobility. The paper \cite{als} introduces bonds for prismatic joints and the paper
\cite{nawratil} introduces bonds for Stewart platforms. The theory can also be used
if the mobility is bigger than one; however, in this paper we will focus on mobility one linkages.

The main purpose of this paper is to make the theory more accessible, because we think that there
is still potential to derive new results on paradoxically moving linkages. The paper is
therefore a survey on bond theory with a tutorial ambition.

In Section~1, we recall
a well-known isomorphism between the Euclidean group $\SE$ of direct isometries from $\R^3$ into itself
and the quotient group of dual quaternions with nonzero real norm by the subgroup of nonzero
real scalars. We use the language of dual quaternions to formulate configuration spaces
of linkages and the closure equations. In Section~3, we recall a well-known description of
all paradoxically moving closed $n$R-loops with revolute joints for $n=3,4,5$; 
the classification of paradoxically moving closed 6R-loops is an open problem which will be
the main question addressed in the subsequent sections. Bonds for R- and P-joints are introduced in Section~3;
this section also contains properties that translate into geometric conditions on the
Denavit/Hartenberg parameters of the linkage under consideration and a proposition illustrating
the use of bond theory for showing non-trivial (but known)
geometric conditions for paradoxically moving loops of type PRRRR. Section~4 introduces
the bond diagram of a linkage, which is useful for ``reading off'' the degree of various 
coupler motions. The last section summarizes the current knowledge on paradoxically moving
6R-loops and their bond diagrams and points out open questions. It also contains the single
new result of this paper (Example~\ref{ex:333}): by specializing a line symmetric linkage,
one may obtain a linkage with three additional rotations in its configuration set.

This article has been accepted for publication in the IMA Journal of Mathematical
Control and Information published by Oxford University Press.

\section{Dual Quaternions and the Closure Equations}

The algebra $\D\H$ of dual quaternions is defined as the 8-dimensional vector space over $\R$
with basis $(1,\qi,\qj,\qk,\eps,\eps\qi,\eps\qj,\eps\qk)$. The multiplication is defined in the
usual way for quaternions; the symbol $\eps$ is supposed to commute with all quaternions and
to fulfill the equation $\eps^2=0$.
Every dual quaternion $h$ can be written as $h=p+\eps q$ with quaternions $p,q$ called
the primal and dual part of $h$. Alternatively, we may write $h$ as an expression
$h=a_0+a_1\qi+a_2\qj+a_3\qk$ with coefficients $a_0,a_1,a_2,a_3\in\D:=\R\oplus\eps\R$
in the ring of dual numbers. The conjugate of $h$ is defined as $\overline{h}:=a_0-a_1\qi-a_2\qj-a_3\qk$;
conjugation is an anti-automorphism of $\D\H$. 


The norm of the dual quaternion is defined as $N(h):=h\overline{h}$. Norm is a 
homomorphism of the semigroup $(\D\H,\cdot)$ to the semigroup $(\D,\cdot)$. 
Let ${\mathbb S}:=\{ h\in\D\H \mid N(h)\in\R \}$ and ${\mathbb S}^\ast:=\{h\in\D\H \mid N(h)\in\R^\ast\}$,
where $\R^\ast:=\R \setminus\{0\}$. Then ${\mathbb S}^\ast$ is a group and $\R^\ast$ is
a normal subgroup. The quotient group ${\mathbb S}^\ast/\R^\ast$ is isomorphic to the
group $\SE$, see \cite[Section~9.3]{selig05}; we may consider it as a locally closed subset $G$ in $\P^7$. Its closure
is the Study quadric $S$, represented by all dual quaternions in ${\mathbb S}$.
We also introduce the \emph{null cone $Y$} defined by the quadratic
form $p+\eps q \mapsto p\overline{p}$.
For instance, if $t\in\R$, then $[t-\qi]$ corresponds to a rotation around the first axis by 
an angle $2\mathrm{arccot}(t)$, and $[1-\eps\qi t]$ corresponds to a translation by 
a distance $2t$ in the direction of the first axis.

\begin{Remark} \label{rem:inf}
In order to parametrize the full rotation group around a fixed axis, we choose the parameter $t$ in
$\P^1_\R = \R\cup \{\infty\}$. For any dual quaternion $h$, the element $(\infty-h)$ is not finite but
we still can consistently say that the class $[\infty-h]$ is equal to $[1]$. The corresponding group element
is the identity.
\end{Remark}

A {\em linkage} is a collection of rigid bodies, called {\em links}, where two links may be 
connected by a {\em joint}. A joint restricts the relative position of the joined links. 
We consider two types of joints:
\begin{enumerate}
\item (R) revolute joints: allow rotations around a fixed axis;
\item (P) prismatic joints: allow translations in a fixed direction;
\end{enumerate}

The {\em link graph} of a linkage is defined by putting a vertex for each link, and an edge
whenever two links are joined by a joint. In order to specify the linkage completely, it suffices
to specify the allowed subset of $\SE$ for each joint. 

Relative positions can be composed by the group operation: the relative position of link~1 with
respect to link~2 times the relative position of link~2 with respect to link~3 is equal to the
relative position of link~1 with respect to link~3. By multiplying relative positions in
a cycle in the link graph, we get the {\em closure equations}. The solution set of the
closure equations is the {\em configuration set} of the linkage.

\begin{Example} \label{ex:6rhom}
Let $n\ge 3$ be an integer.
The link graph of a closed $n$R linkage is an $n$-cycle. For $r=1,\dots,n$, the set of allowed relative
position can be written as $\{ [t_r-h_r] \mid t_r\in\R\cup\{\infty\} \} = \{ [t_r-h_r] \mid t_r\in\R \} \cup \{[1]\}$, where $h_r$ is a dual
quaternion such that $h_r^2=-1$ specifying the rotation axis in an initial position. 
This gives the closure equation
\begin{equation}\label{eq:ch}
 [(t_1-h_1)(t_2-h_2)\dots(t_n-h_n) ] = [1] . 
\end{equation}
The class on the left hand side is $[1]$ if and only if 7 of the 8 coordinates of the product
are zero, hence we have 7 polynomial equations in $t_1,\dots,t_n$. Actually, one of the 7 equations
is redundant, because it is clear that the product is contained in the Study quadric $S$; and
if $h=a_0+a_1\qi+a_2\qj+a_3\qk$, $a_0,\dots,a_3\in\D$, is a dual quaternion with norm in $\R^\ast$ 
such that $a_1=a_2=a_3=0$, then it follows that $h\in\R^\ast$.

Assume $n=6$.
For generic choice of $h_1,\dots,h_6\in\D\H$ such that $h_r^2=-1$ for $r=1,\dots,6$, one gets
at most 16 real solutions for these 6 equations in $t_1,\dots,t_6$, including the solution
$t_r=\infty$ for $r=1,\dots,6$. The number of complex solutions
is infinite, but if one excludes solutions contained in the null cone $Y$,
then one generically gets 16 complex solutions \cite[pp.~262--264]{selig05}.
\end{Example}

\begin{Example} \label{ex:6rinhom}
In the case of a closed $n$R linkage (or more general simply closed linkages), it is possible
to use a more invariant specification which does not depend on the choice of an initial position.
In essence, this is the method of Denavit/Hartenberg~\cite{Denavit}.
For $r=1,\dots,n$, let $\phi_r$ be the angle between the $r$-th and $(r+1)$-th rotation axis,
with indices modulo $n$; let $d_r$ be the normal distance between these axes; let $s_r$ be the
signed distance of the intersections of the common normals of neighboring axes on the $r$-th
rotation axis. For each linkage, we introduce an internal frame of reference in which the first
joint is the first coordinate axes, the intersection with the common normal with the previous
axes is the origin, and the second axis lies in the first coordinate plane. When we express relative
positions in this frame of reference, then the allowed positions are composed by a translation
in direction of the first axes by $s_r$, rotation around first axis by $\phi_r$, translation
in direction of the second axis by $d_r$, and rotation around the second axis by an arbitrary angle.
So the closure equation is
\begin{equation}\label{eq:cinh}
  [(t_1-\qi)g_1(t_2-\qi)g_2\cdots(t_n-\qi)g_n] = [1],
\end{equation}
where
\begin{equation}\label{gi:1}
g_r=\left(1-\frac{s_r}{2}\eps\qi\right)\left(w_r-\qk\right)\left(1-\frac{d_r}{2}\eps\qk\right)
\end{equation}
and $w_r=\cot(\frac{\phi_r}{2})$ for $r=1,\dots,n$. If $\phi_r$ is a multiple of $\pi$ for some $r$,
then we set $w_r=\infty$ and 
$g_r=\left(1-\frac{s_r}{2}\eps\qi\right)\left(1-\frac{d_r}{2}\eps\qk\right)$.

Similar as in the previous example, we get 6 equations in $n$ parameters, and if $n=6$, 
then for generic choice
of the parameters $s_1,w_1,d_1,\dots,s_6,w_6,d_6$ we get 16 isolated complex solutions.
(But now there is no trivial solution at infinity.)
\end{Example}

\begin{Remark}
The invariant parameters for a closed $n$R linkage above do depend on a choice of orientation
of the rotation axis. Generically, there are $2^n$ choices leading to different parameters.
If we change the orientation of the $k$-th axis, then $w_k$ gets replaced by $-1/w_k$ and $s_k$
gets replaced by $-s_k$, and all other parameters stay the same. Also, the parameter $d_k$ can be
replaced by $-d_k$ without changing the linkage, one just needs to reparametrize the
sets of allowed positions (replacing $t_k$ by $-1/t_k$).

If $w_k=0$ for some $k$, then one can add a constant to $s_k$ and $s_{k-1}$ without changing
the linkage, because then the axes are parallel and the common normal is not unique. Similarly,
if $w_k=\infty$, then one can add a constant to $s_k$ and subtract the constant from $s_{k-1}$.
\end{Remark}

\begin{figure}
  \centering
  \includegraphics{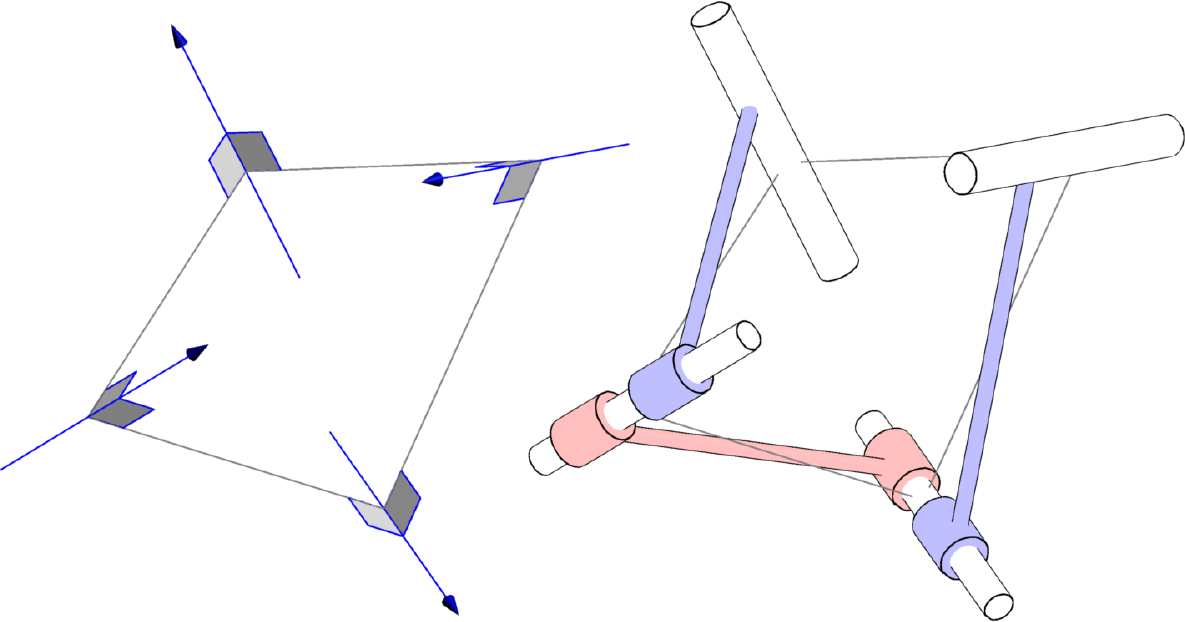}
  \caption{Bennett linkage}
  \label{fig:bennett}
\end{figure}

\begin{Example} \label{ex:bennett}
With the notation from Example~\ref{ex:6rinhom}, we set $n=4$ and
\[ (w_1,w_2,w_3,w_4) = (1,2,1,2) , \]
\[ (d_1,d_2,d_3,d_4) = (4,5,4,5) , \]
\[ (s_1,s_2,s_3,s_4) = (0,0,0,0) . \]
The closure equation~\eqref{eq:cinh} for $(t_1,t_2,t_3,t_4)$ can be simplified using computer algebra (we used Maple).
The simplified system is 
\[ 3t_1t_2+1 = t_1+t_3 = t_2+t_4 = 0 , t_1^2+1\ne 0, t_2^2+1\ne 0. \]
Its solution set is a curve that can be easily parametrized; it is
\[ (t_1,t_2,t_3,t_4) = \left( t,\frac{-1}{3t}, -t, \frac{1}{3t} \right) . \]
This linkage is therefore mobile. It is an example of a Bennett linkage \cite{bennett14} (Figure~\ref{fig:bennett}). 
The general description of Bennett linkages is given by the equations
\begin{equation}
  \label{eq:1}
  \begin{gathered}
s_1=s_2=s_3=s_4=0,\ w_1=w_3\ne 0,\infty,\ w_2=w_4\ne 0,\infty\\
d_1=d_3\ne 0,\ d_2=d_4\ne 0,\ 
 \frac{2d_1w_1}{w_1^2+1}=\frac{2d_2w_2}{w_2^2+1} ,
  \end{gathered}
\end{equation}
which lead to a similar one-dimensional solution set.
\end{Example}

\section{Mobile Closed 4R and 5R Linkages}

For $n=3,4,5$, and for generic choice of parameters, the closure equation~\eqref{eq:ch} has 
only the trivial solution, and the closure equation~\eqref{eq:cinh} has no solution at all
(disregarding the complex solutions with factors in the null cone $Y$). Nevertheless, there are
special cases of mobile linkages of these types. Our main question is what are the implications
of the assumption of mobility for the structure of a linkage. The hope is to find enough conditions
that allow a classification of mobile loops with 4,5 (and later 6) links.

\begin{Example} \label{ex:ca}
If $h_1=h_s$ for some $s, 2\le s\le n$, then the closure equation~\eqref{eq:ch} has the solution
\[ t_s=-t_1, t_r=\infty \mbox{ for } r\ne 1,s. \]
\end{Example}
These are configurations where two of the axes coincide, and one part of the linkage just
rotates about the coinciding axes.

It is easy to show that every mobile 3R linkage is of this type. In case all 3 axes coincide,
the mobility is 2.

\begin{Example} \label{ex:sph}
Assume that $n=4$. Assume that the dual parts of $h_1,\dots,h_4$ are all zero. Then the
dual part of the left hand side of equation~\eqref{eq:ch} is automatically zero and the
closure equation boils down to three polynomial equations in $t_1,t_2,t_3,t_4$. In general,
the solution has complex dimension~1. 

Geometrically, the vanishing of the dual part means that all four axes pass through the origin.
This type of linkage is known as {\em spherical 4R linkage}. It can also be
characterized by the conditions
\[ s_1=s_2=s_3=s_4=d_1=d_2=d_3=d_4=0 . \]

A similar case (which may actually be considered as limiting case of a spherical linkage) is the
{\em planar 4R linkage}, where all 4 axes are parallel.
\end{Example}

It is well-known \cite{delassus} that every mobile 4R linkage either has two coinciding axes
or is spherical, planar, or Bennett. 

\begin{Example} \label{ex:frozen}
Let $n=4$. Let $h_1,h_2,h_3,h_4$ be the dual quaternions defining the rotation axes of
a mobile 4R linkage (e.g. spherical). Choose an arbitrary dual quaternion $h_5$ such that
$h_5^2=-1$. Then any configuration of the mobile 4R linkage can be extended to a configuration
of the 5R linkage by setting $t_5:=\infty$. Hence the 5R linkage is again movable.

The 5th joint in this linkage remains frozen during this particular motion.
\end{Example}

\begin{figure}[ht]
\begin{center}
\includegraphics{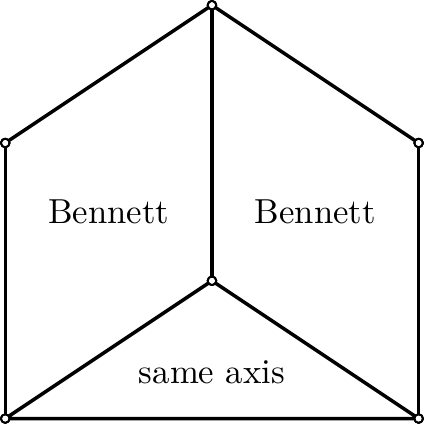}
\end{center}
\caption{The link graph of a triply closed linkage with 6 links and 8 joints (three joints with
coinciding axes) of mobility 1. If the link in the center is removed,
	one obtains a Goldberg 5R linkage.}
\label{fig:gblg}
\end{figure}

\begin{Example} \label{ex:goldberg}
Consider the triply closed linkage with 6 links that are connected according to the link graph
in Figure~\ref{fig:gblg}. Its mobility is one: the motion of the joint corresponding 
to the vertical edge in the middle determines the motion of the two Bennett 4R linkages, and then
the motion of the 3R linkage below is also determined. If we remove the link in the center,
then we get a mobile closed 5R linkage. This construction is due to Goldberg \cite{goldberg}.
\end{Example}

It is well-known that every mobile closed 5R linkage either has coinciding axes, or a frozen joint,
or is planar, spherical, or a Goldberg linkage. The original proof \cite{karger} uses
computer algebra; a simpler proof without computers is based on bond theory \cite{hss2}.

Let ${\cal L}_n:=\R^{2n}\times (\P^1)^n$ denote the parameter space of closed $n$R-linkages. 
Let $D_n:={\cal L}_n\times (\P^1_\C)^n$ be the Zariski closure of the set of all solutions 
of equation~\eqref{eq:cinh}. The projection $D_n\to{\cal L}_n$ is a proper morphism, 
hence the dimension of its fiber is upper semicontinuous in the Zariski topology 
as a function in ${\cal L}_n$. The subset ${\cal M}_n\subset {\cal L}_n$ of all parameters 
of mobile linkages is then also Zariski closed, i.e., it is a subset defined by algebraic
equations in the parameters $s_1,w_1,d_1,\dots,s_n,w_n,d_n$. The linkages with parameters in ${\cal M}_n$
have an infinite solution set over the complex numbers. It is possible that none of these
solutions are real, for instance for a planar linkage where $d_1>d_2+d_3+d_4$. 

For $n=3,4,5$, these equations are known (for $n=3,4$, we essentially gave them above; for $n=5$, see \cite{dietmaier}). 
The dimension of ${\cal M}_n$
is determined by the largest components, corresponding to linkages with coinciding/parallel axes
($\dim({\cal M}_3)=3$, $\dim({\cal M}_4)=7$, $\dim({\cal M}_5)=10$); but the more interesting components have
smaller dimension (Bennett linkages form a 3-dimensional component of ${\cal M}_4$, and
Goldberg linkages form a 5-dimensional component of ${\cal M}_5$). For ${\cal M}_6$,
we do know the dimension: it is 14, again because of components with coinciding axes. 
Several other components are known and will be discussed in the following. Even more
algebraic subsets of ${\cal M}_6$ are known which are not contained in any known component,
but for which it is not clear whether they form a component or they are properly contained
in some yet unknown component. It is an open problem
to determine all components and to give equations for them. A broad discussion of partial
results can be found in \cite{li}.

\section{Definition and First Properties of Bonds}

Assume that we have a linkage with $e$ joints. If the $k$-th joint is of type R, then the
set of allowed motions can be parametrized by $t_k\mapsto m_k:=(t_k-h_k)g_k$ for some $h_k,g_k\in\D\H$
of norm 1 with $h_k^2=-1$, with $t_k\in\R$. If we pass to classes, we may even allow $t_k\in\P^1_\R$,
see Remark~\ref{rem:inf}.
The parametric dual quaternion $m_k$ also appears as a factor in closure
equations. Note that $N(m_k)=t_k^2+1$.

If the $k$-th joint is of type P, then the set of allowed motions can be parametrized by
$t_k\mapsto m_k:=(t_k-\eps p_k)g_k$ for some purely vectorial $p_k\in\H$ (that is, $p_k + \overline{p_k} = 0$) and $g_k\in \D\H$, both of norm 1, with
$t_k\in\R\setminus \{0\}$. Here, we have $N(m_k)=t_k^2$.

Recall that the configuration set $K$ is the set of all $(t_1,\dots,t_e)$ such that
$m_{i_1}\dots m_{i_r}\in\R^\ast$ for all loops with edges $i_1,\dots,i_r$. In order to define
bonds, we have to allow also complex parameters $t_1,\dots,t_e$. Recall that the Zariski
closure of any set $X$ is defined as the set of all points, maybe with complex coordinates,
which satisfy all polynomial equations that are satisfied by all points in $X$.
The Zariski closure of $K$ in the product of complex projective lines 
is denoted by $\overline{K}$. The set $B$ of {\em bonds} is defined as the set of
all elements $(t_1,\dots,t_e)\in\overline{K}$ such that at least one of the $m_k(t_k)$ has
norm zero.
For a fixed bond, the subset of joints $k$ such that $N(m_k(t_k))=0$ are called the joints
attached to the bond. The bond also induces a partition of all links into the connected components
of the subgraph which is obtained by deleting all attached joints. 

The following proposition guarantees that the set of bonds of a mobile linkage
is non-empty.

\begin{Proposition} \label{thm:frozen}
A joint is frozen in a linkage if and only if it is not attached to any bond.
\end{Proposition}

\begin{proof}
If the $k$-th joint is frozen, then $t_k=c$ for some constant $c$ for all configurations
in $K$. Hence we also have $t_k=c$ for all points in $\overline{K}$, which includes all bonds.
Hence $N(m_k)\ne 0$.

Conversely, if the $k$-th joint is not attached to any bond, then the projection from $\overline{K}$
to the coordinate $t_k$ is not surjective. On the other hand, the image is a closed subvariety.
It follows that the image is finite and therefore the joint is frozen.
\end{proof}

\begin{Example} \label{ex:bondbenett}
The Zariski closure of the configuration set of the Bennett linkage in Example~\ref{ex:bennett} is
\[ \{ (t_1,t_2,t_3,t_4) \mid 3t_1t_2+1 = t_1+t_3 = t_2+t_4 = 0 \}. \]
It has 4 bonds:
\[ \beta_1 := (\ci,-\ci/3,-\ci,\ci/3),\ \beta_2 := (-\ci,\ci/3,\ci,-\ci/3), \]
\[ \beta_3 := (\ci/3,-\ci,-\ci/3,\ci),\ \beta_4 := (-\ci/3,\ci,\ci/3,-\ci). \]
The joints 1 and 3 are attached to $\beta_1$ and $\beta_2$. 
The joints 2 and 4 are attached to $\beta_3$ and $\beta_4$. 
\end{Example}

\begin{Proposition}
Let $i_1,\dots,i_k$ be a set of joints forming a path in the link graph such that starting point
and ending point are in the same subset of the partition induced by the bond $\beta = (t_1,\ldots,t_e)$. If at least one
of the joints is attached to the bond $\beta$, then $m_{i_1}(t_{i_1})\cdots m_{i_k}(t_{i_k})=0$.
\end{Proposition}

\begin{proof}
Since the starting and ending point are in the same subset of the partition, there exists a path $(j_1,\dots,j_l)$ of joints
not attached to the bond $\beta$, with the same starting and ending point. The closure equation
\[ [m_{i_1}\cdots m_{i_k}] = [m_{j_1}\cdots m_{j_l}]  \]
is valid for all configurations, but it is not valid for $\beta$ because the left side has norm zero and the right
side has norm different from zero. Since $\beta$ is in the closure of the configuration set, there is only one possibility:
the left side is not defined, because the product is zero.
\end{proof}

\begin{Example} \label{ex:bennettbc}
In our running example of the Bennett linkage (Example~\ref{ex:bennett} and Example~\ref{ex:bondbenett}), 
the following equations
and their conjugate counterparts obtained by replacing $\ci$ by $-\ci$ are valid:
\[ (\ci-\qi)g_1(-\ci/3-\qi)g_2(-\ci-\qi)g_3 = 0 , \]
\[ (\ci-\qi)g_2(\ci/3-\qi)g_3(-\ci-\qi)g_4 = 0 , \]
\[ (\ci-\qi)g_3(-\ci/3-\qi)g_4(-\ci-\qi)g_1 = 0 , \]
\[ (\ci-\qi)g_4(\ci/3-\qi)g_1(-\ci-\qi)g_2 = 0. \]
\end{Example}

Assume that we have a minimal chain $i_1,\dots,i_k$ such that $m_{i_1}(t_{i_1})\cdots m_{i_k}(t_{i_k})=0$
for some fixed bond $(t_1,\dots,t_e)$. Then $N(m_{i_1}(t_{i_1}))=N(m_{i_k}(t_{i_k}))=0$ - otherwise we could
multiply by $\overline{m}_{i_1}(t_{i_1})$ from the left or by $\overline{m}_{i_k}(t_{i_k})$ from the right and make the chain
shorter. The condition
\begin{equation} \label{eq:bond}
 m_{i_1}(t_{i_1})\cdots m_{i_k}(t_{i_k})=0,\ N(m_{i_1}(t_{i_1}))=N(m_{i_k}(t_{i_k}))=0 
\end{equation}
is called {\em bond condition}. The validity of a bond condition for some chain in the link graph
has some interesting geometric consequences on the geometric parameters of the linkage.

\begin{Lemma} \label{lem:bcon1}
\begin{enumerate}
\item If joint 1 is of type R and joint 2 is of type R or P, then the bond condition 1-2 is never valid.
\item If joint 1 is of type P and joints 2 and 3 of type R, and the bond condition 1-2-3 is valid,
	then the axes of joints 2 and 3 are parallel.
\end{enumerate}
\end{Lemma}

\begin{proof}
(1): Assume that joints 1 and 2 are of type R, and the axis have distance $d$ and twist angle $2\mathrm{arccot}(w)$.
Assume, without loss of generality, that the bond coordinates at joints 1 and 2 are both $\ci$ and not $-\ci$
(this can always be achieved by a change of orientation of the axes). With 
\[ g= \left(w-\qk\right)\left(1-\frac{d}{2}\eps\qk\right), \]
the bond condition reduces to the equation
\[ (\ci-\qi)g(\ci-\qi)= (2w\ci-\eps d\ci)(\ci-\qi) = 0 ,\]
hence $w=d=0$ and the axes are equal.

Assume now that joint 2 is of type P. Then the bond coordinate at the second joint is $0$, and the bond condition
has the form $(t_{1}-\qi)\eps p=0$ for some quaternion $p$ specifying the direction of the P-joint. Since $p$ is invertible,
it follows that $t_1-\qi=0$, which is impossible.

The statement (2) reduces to a similar short and straightforward calculations.
\end{proof}

For any chain $c:=(i_1,\dots,i_k)$, the {\em coupling space} $L_c$ is defined as the linear subspace 
of $\D\H$ generated by all products $m_{i_1}(t_{i_1})\cdots m_{i_k}(t_{i_k})$, where $t_{i_1},\dots,t_{i_k}$ 
range over the full parameter space. If all joints are of type R or P, then $L_c$ has a 
generating set of cardinality $2^k$ which can be obtained by expanding the product and taking
all coefficients with respect to $t_{i_1},\dots,t_{i_k}$. The projectivization of the coupling
spaces contains the {\em coupling varieties}, consisting of all relative positions of the two links
that are connected by the chain.

\begin{Lemma} \label{lem:bcon2}
Let $c:=(i_1,\dots,i_k)$ be a chain of joints.
\begin{enumerate}
\item If the joint $i_1$ or the joint $i_k$ is of type R, then $\dim(L_c)$ is even.
\item If $\dim(L_c)=2$, then all joints are of the same type and have the same axis (for R-joints)
	resp. directions (for P-joints).
\item If all joints are of type $R$ and $\dim(L_c)=4$, then all axes are parallel or pass through a
	common point.
\item If $k=3$ and there is a bond such that the bond condition for $c$ is valid, then $\dim(L_c)<8$.
\end{enumerate}
\end{Lemma}

\begin{proof}
This is \cite[Theorem 1]{hss2}. The proofs of (1) and (4) do give some insight, so we include them here.

(1): assume that $i_1$ is an R-joint. Then $L_c$ is closed under multiplication by $h_1$ from the left.
Since left multiplication by $h_1$ is a linear map whose square is negative identity, it follows that $L_c$
may be considered as a vector space over $\C$. Its real dimension is two times its dimension over $\C$,
which proves the claim.

(4): Expanding the bond condition, we obtain a nontrivial linear equation between the products generating
$L_c$. Hence these products cannot be linearly independent.
\end{proof}

Assume that we have a chain 1-2-3 of three joints of type R. Then we say that the chain satisfies the
{\em Bennett condition} if $\dim(L_c)=6$. The following proposition expresses the condition
in terms of the Denavit/Hartenberg parameters. The proof is straightforward.

\begin{Lemma} \label{lem:bcon3}
Let $d_1,d_2,w_1,w_2,s_2$ be the distances, angles, and offset of a 3-chain of R-joints.
Then the Bennett condition is equivalent to
\[ s_2=0,\ \frac{2d_1w_1}{w_1^2+1}=\frac{2d_2w_2}{w_2^2+1} \]
(compare with Equation~\eqref{eq:1}).
\end{Lemma}

The lemmas ~\ref{lem:bcon1}, \ref{lem:bcon2},and \ref{lem:bcon3} above give necessary conditions on the geometric
parameters of a linkage for the existence of bonds.
The following proposition, which is taken from \cite[Theorem 6]{als}, demonstrates how these lemmas
are applied to classify linkages with a given link diagram and types of joints.



\begin{Proposition} \label{ex:prrrr}
Consider a mobile closed PRRRR linkage. Then one of the following conditions must be satisfied.
\begin{enumerate}
\item The P-joint is frozen.
\item Two of the axes of the rotational joints coincide.
\item Three of the rotational axes are parallel, and the fourth axes is frozen.
\item The axes of joints 2 and 3 are parallel, and the axes of joints 4 and 5 are parallel.
\end{enumerate}
\end{Proposition}

\begin{proof}
By Lemma~\ref{lem:bcon1}, we get that either joints 2 and 3 are parallel {\em or} joints 4 and 5 are parallel.
In order to show that actually both are necessary, we consider the closure equation modulo $\eps$.
If, say, joints 2 and 3 are parallel, and joints 4 and 5 are not, then there are three different rotation axes, since parallel
axes only differ in their dual part. Modulo $\eps$, we get a closure equation of a 3R loop with three different
axes, but such a link is never movable.
\end{proof}

%

For chains of four R-joints, the bond condition does not imply geometric conditions on the four axes.
Indeed, given four generic dual quaternions $h_1,h_2,h_3,h_4$, there are two solutions of the equation
\[ (\ci-h_1)(t_2-h_2)(t_3-h_3)(\ci-h_4) = 0 \]
in the unknowns $t_2,t_3$. This gives room for two bonds with $t_1=t_4=\ci$. Varying the signs of $t_1,t_4$,
one has up to 8 choices for the coordinates $(t_1,t_2,t_3,t_4)$ of a bond in such a chain.

In a 6R loop with joints 1-2-3-4-5-6-1, one can get a geometric condition by comparing the possible solutions of bond coordinates
in the chains 1-2-3-4 and 4-5-6-1. The idea is to take into account the equality
\[ [(t_1-h_1)(t_2-h_2)(t_3-h_3)] = [(t_6+h_6)(t_5+h_5)(t_4+h_4)] , \]
which holds for all points in the configuration set, and therefore also for the bonds in case both sides are defined.
The lengthy calculations have been done in \cite{ls} for 6R linkages given in terms of their Denavit/Hartenberg parameters
$d_1,w_1,s_1,\dots.d_6,w_6,s_6$. For $i=1,\dots,6$, we define $c_i:=\frac{w_i^2-1}{w_i^2+1}$ and $b_i:=\frac{2d_iw_i}{w_i^2+1}$.
Then we define the {\em quad polynomial} as the quadratic polynomial in a variable $x$
\[ Q_1^+(x) = \left(x+\frac{b_{3}c_3-b_{1}c_1}{2}-\frac{s_{1}}{2}\ci\right)^2 + \]
 \[ \frac{\ci}{2}\left(b_1 s_{2}+b_{3} s_{3}+s_{2} b_{3} c_{2}+s_{3} b_1 c_{2}\right) -\]
 \[ \frac{b_1 b_{3} c_{2}-s_{2} s_{3} c_{2}}{2}
   + \frac{s_{2}^2+s_{3}^2-b_1^2+b_{2}^2-b_{3}^2-b_{2}^2 c_{2}^2}{4}. \]
For $i=2,\dots,6$, we define the quad polynomial $Q_i^+(x)$ by a cyclic shift of indices that
shifts $1$ to $i$. Finally, we define $Q_i^-(x)$ by replacing the parameters $c_1,\dots,c_6,b_1,\dots,b_6$
and $s_2,s_4,s_6$ by their negatives, and leaving $s_1,s_3,s_5$ as they are. For instance,
\[ Q_1^-(x) = \left(x+\frac{b_{3}c_3-b_{1}c_1}{2}-\frac{s_{1}}{2}\ci\right)^2 + \]
 \[ \frac{\ci}{2}\left(b_1 s_{2}-b_{3} s_{3}-s_{2} b_{3} c_{2}+s_{3} b_1 c_{2}\right) -\]
 \[ \frac{-b_1 b_{3} c_{2}-s_{2} s_{3} c_{2}}{2}
   + \frac{s_{2}^2+s_{3}^2-b_1^2+b_{2}^2-b_{3}^2-b_{2}^2 c_{2}^2}{4}. \]

\begin{Theorem} \label{thm:q}
Let $k$ be the number of bond connections of $1$ and $4$. Then
\[ k \le \deg(\gcd(Q_1^+,Q_4^+)) + \deg(\gcd(Q_1^-,Q_4^-)) . \]
\end{Theorem}

One can use this to derive a necessary condition for the existence of a bond connecting 1 and 4, since two polynomials
have a nontrivial $\gcd$ if and only if their resultant is zero. The equations get simpler if one assumes that some of the $\gcd$'s
have degree~2, because then the two quad polynomials need to be equal.

Similar conditions should be possible for 6-loops with P-joints, but the equations have not yet been derived, and so
its consequences are not yet known.

\section{Bond Diagrams}

The combinatorial structure of the bonds -- how many bonds are attached to which joints -- can be visualized
in a diagram. This diagram can be used to read off the degree of coupling curves.

Consider a linkage of mobility 1 with $e$ joints. Let $K$ be its configuration set and let 
$\beta=(\beta_1,\dots,\beta_e)\in\overline{K}$ be a bond ($\overline{K}$ denotes the Zariski closure of $K$).
For any two links $i,j$, define the {\em coupling map} $f_{i,j}:\overline{K}\to\SE$ as the map that computes
the relative position of link $j$ with respect to the link $i$. 
We define the {\em local distance} $d_\beta(i,j)$ as the order of the Taylor expansion
of the analytic function $N\circ f_{i,j}\circ \lambda$ divided by~2, where $\lambda$ is a local parametrization of $\overline{K}$
around~$\beta$.

\begin{Example} \label{ex:ldist1}
In the Bennett linkage of Example~\ref{ex:bennett} whose bonds are given in Example~\ref{ex:bondbenett}, we consider the bond
$\beta_1=(\ci,-\ci/3,-\ci,\ci/3)$. A local parametrization $\lambda$ of $\overline{K}$ is 
$(t-\ci,(-t+\ci)/3,-t+\ci,(t-\ci)/3)$.
We label the links cyclically, the link attached to joints 1 and 2 is indexed by 2.
Then $f_{4,2}\circ \lambda$ is $(t-\ci-\qi)g_1((-t+\ci)/3\qi)g_2$, and its norm has Taylor expansion
\[ N((t-\ci-\qi)N(g_1)N((-t+\ci)/3-\qi)N(g_2)= -160/9\ci t+40/9 t^2 + {\cal O}(t^3), \]
which is of order 1. Hence $d_{\beta_1}(4,2)=1/2$.

For the other links, one gets $d_{\beta_1}(4,1)=d_{\beta_1}(1,4)=d_{\beta_1}(2,3)=d_{\beta_1}(3,2)=1/2$, 
and all other values of $d_{\beta_1}$ are zero.
\end{Example}

\begin{Lemma} \label{lem:ldist1}
The local distance is a pseudo-metric on the set of links, i.e. for any three links $i,j,k$ we have
\[ d_\beta(i,i)=0, \ d_\beta(i,j) = d_\beta(j,i), \] 
\[ d_\beta(i,k) \le d_\beta(i,j) + d_\beta(j,k) . \]
Moreover, the perimeter $d_\beta(i,j)+d_\beta(j,k)+d_\beta(k,i)$ of the triangle $(i,j,k)$ is even.
\end{Lemma}

\begin{Lemma} \label{lem:ldist2}
Assume that the link graph contains a chain from $i$ to $k$ passing $j$, for which the bond condition imposed by $\beta$
is not valid. Then 
\[ d_\beta(i,k) = d_\beta(i,j) + d_\beta(j,k) . \]
\end{Lemma}

\begin{proof}
Both Lemmas are consequences of \cite[Theorem 3]{hss2}. The necessary adaptions to include prismatic joints are easy.
\end{proof}

The vast majority of bonds we studied so far have a very simple local distance: the link graph is partitioned into two
subsets; $d_\beta(i,j)=\frac{1}{2}$ if $i$ and $j$ lie in different subsets, and $d_\beta(i,j)=0$ otherwise
(see e.g. Example~\ref{ex:ldist1} above). We visualize linkages with these bonds by adding to the link diagram
additional lines connecting the edges, one for each conjugated pair of bonds, 
that separate the vertices in the same way. Figure~\ref{fig:bbbg}
shows the bond diagrams of the Bennett 4R linkage and 
of the Goldberg 5R linkage.

\begin{figure}
\begin{center}
\includegraphics{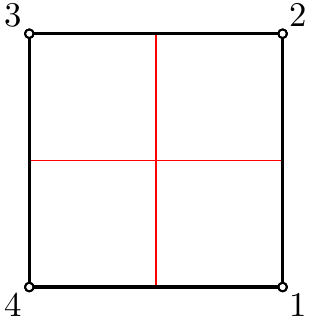}\quad\quad
\includegraphics{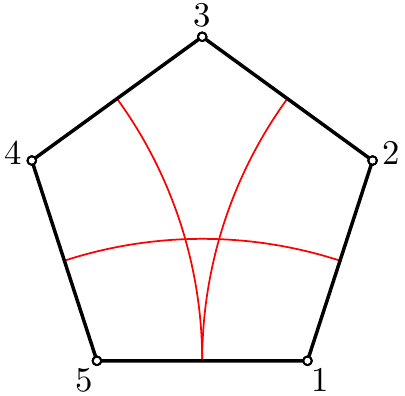}
\end{center}
\caption{The bond diagrams of the Bennett 4R linkage and of the Goldberg 5R linkage. 
	As explained in Exercise~\ref{ex:gbb}, one can ``read off'' the degree of its various
	coupler curves, be counting the number of lines that cross a line connecting to vertices corresponding
	to links. 
	For instance, the coupling curves $C_{14}$ in the Goldberg linkage is a cubic.}
\label{fig:bbbg}
\end{figure}

For any two links $i,j$, the algebraic degree of the coupling curve $C_{i,j}$ is defined as the number of all points 
$p\in\overline{K}$ such that $f_{i,j}(p)$ lies in a fixed generic hyperplane of $\P^7$. If $f_{i,j}:\overline{K}\to C_{i,j}$
is birational, then this is simply the degree of $C_{i,j}$ as a curve in $\P^7$. In general, it is equal to the degree
of $C_{i,j}$ multiplied with the mapping degree of $f_{i,j}$, i.e. the number of preimages of a generic point of $C_{i,j}$.

\begin{Theorem} \label{thm:adeg}
For any two links, the algebraic degree of $C_{i,j}$ is equal to the sum of all local distances $d_\beta(i,j)$ over
all bonds $\beta$.
\end{Theorem}

\begin{proof}
  For R joints, this is \cite[Theorem 4]{hss2}. The adaptions to make it work for P joints are easy to make. The idea of the proof is the following: 
  instead of computing the algebraic degree
by intersection with a hyperplane, we can also intersect with a quadric not containing any coupler curve. We take
the null cone $Y$. Then we do the counting taking multiplicities into account, and finally divide by 2. 
\end{proof}

In the bond diagram, the local distance of a bond visualized by a single line separating the link diagram into
two is $0$ for any two vertices in the same subset and it is $\frac{1}{2}$ for any two vertices in different subsets. Summing up over all bonds, and taking into account
that conjugate bonds separate the same links, we obtain that the algebraic degree of a coupling curve is equal to
the number of lines that must be crossed when one draws a line between the two vertices.

\begin{Example} \label{ex:gbb}
Figure~\ref{fig:bbbg} shows the bond diagram of Goldberg's 5R linkage. The coupling curves $C_{13}$, $C_{35}$,
and $C_{24}$ are conics (the two first actually appear in the Bennett linkage in the construction). The coupling curves 
$C_{14}$ and $C_{25}$ are cubics. The remaining 5 coupling curves
are lines, parametrizing rotations around the joint axes. The coupling map $f_{15}$ is 2:1, all other coupling
maps are birational.
\end{Example}

\begin{figure}
\begin{center}
\includegraphics{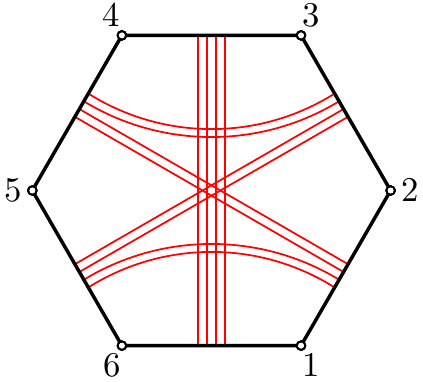} \quad\quad
\includegraphics{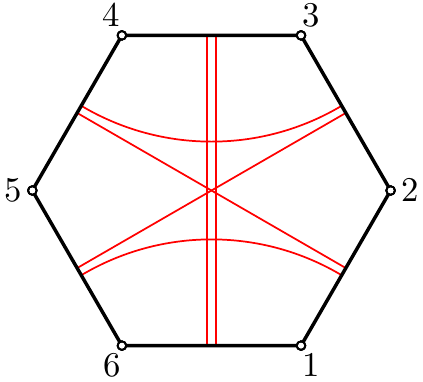}
\end{center}
\caption{The bond diagrams of Dietmaier's 6R linkage and of Wohlhart's double Goldberg linkage. The degree of the coupler curve $C_{25}$ is 8 
  in Dietmaier's linkage and 4 in Wohlhart's linkage
  (see also Example~\ref{ex:dietb} and Example~\ref{ex:wgg}).}
\label{fig:bdiet}
\end{figure}

\begin{Example} \label{ex:dietb}
The mobile 6R linkage found by Dietmaier~\cite{dietmaier} can be characterized by the following condition: the coupling
spaces of two disjoint chains of length 3 both have dimension~6, and intersect in a space of dimension~5. Then the coupling curve
with respect to the two links connected by the two chains 
is contained in the projectivization of the intersection, which is a $\P^4$. It is defined by three quadratic
equations and therefore it has degree~8.

The bond diagram is shown in Figure~\ref{fig:bdiet}. One can see the algebraic degree of $C_{25}$. Also, the Bennett
conditions need to hold for two triples of axes because of the bond connections 1-5 and 2-4.
\end{Example}

\begin{Example} \label{ex:wgg}
In \cite{wgg}, Wohlhart constructed a movable 6R linkage by combing four Bennett linkages (see Example~\ref{ex:compb}
for the construction). Its bond diagram is resembling
the bond diagram of Dietmair's linkage. The dimensions of coupling spaces are the same, and the
number of bonds attached to any pair of joints is exactly half as in Example~\ref{ex:dietb}. Consequently, the degree
of each coupler curve is exactly half of the degree of the corresponding coupler curve in Dietmaier's linkage.
\end{Example}

\begin{figure}
\begin{center}
\includegraphics{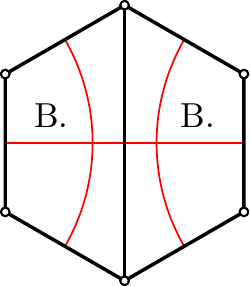}
\includegraphics{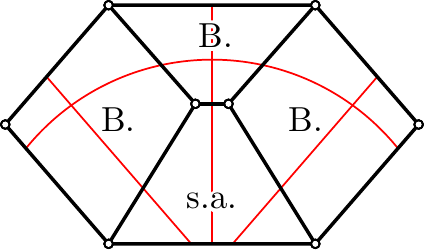}
\includegraphics{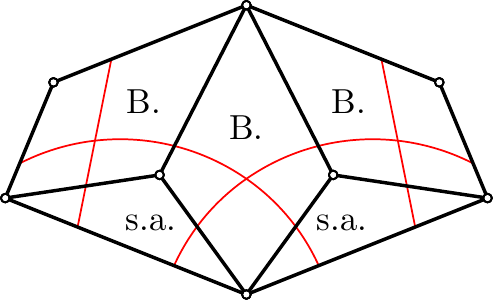}
\includegraphics{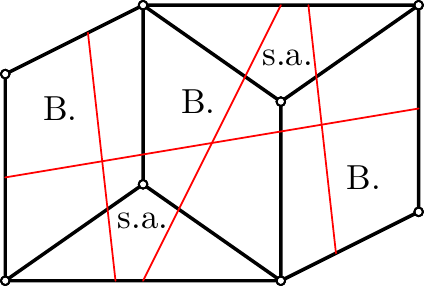}
\includegraphics{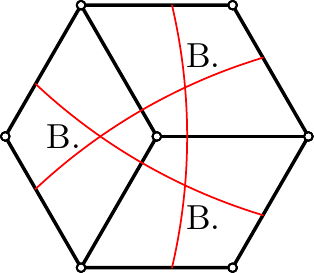}
\includegraphics{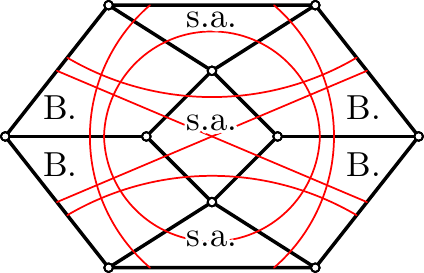}
\end{center}
\caption{Bond diagrams of multiply closed linkages consisting of Bennett 4R linkages and 3R linkages with the same axis (``s.a.'')
	in all three joints, including bond diagrams. If one leaves away the links in the interior, then one obtains
	Waldron's double Bennett linkage, the Goldberg L-form 6R linkage, Goldberg's serial 6R linkage, 
	Goldberg's second serial 6R linkage, the cube linkage, and Wohlhart's double Goldberg linkage.}
\label{fig:compb}
\end{figure}

\begin{Example} \label{ex:compb}
By combining Bennett linkages and 3R linkages such that all 3 axes coincide, one can construct various multiply closed
linkages. Leaving away the linkages with vertices drawn in the interior of the planar representation of the link graph,
we obtain the following 6R linkages: Waldron's double Bennett linkage~\cite{waldron}, the Goldberg L-form 6R linkage,
Goldberg's serial 6R linkage, Goldberg's second serial 6R linkage (all in \cite{goldberg,baker}), the cube linkage~\cite{hss1},
and Wohlhart's double Goldberg linkage~\cite{wgg}. Figure~\ref{fig:compb} shows the bond diagrams
of the multiply closed linkages. 

A complete list of all linkages with three or four conjugate pairs of bonds can be found in \cite{li}.
\end{Example}

\section{Open questions for 6R linkages}

Bond theory has helped to discover new families of 6R linkages (see \cite{hss2,sharp,ls}), sometimes containing
known families (see \cite{hlss}). The possible list of bond diagrams is finite. For some of these diagrams,
we know all linkages, for other diagrams we do have examples but no proof of completeness, and for other possible
diagrams we do not if they appear as the diagrams of any linkage. In this section, we give a summary of
the open cases.

In this overview we exclude degenerate cases where one of the joints is frozen or where two consecutive
axes coincide. Also, we exclude cases where three consecutive axes are incident to a single
point; then the three joints can be replaced by a spherical joint, and the mobile SRRR linkages are
well understood. We also exclude cases with three consecutive axes being parallel, since they can be seen
as limit cases of three consecutive axes that are incident to a single point. By Lemma~\ref{lem:bcon2},
the dimension of the coupling spaces $L_{i,i+1,i+2}$ is either 6 or 8. We label the links and joints cyclically
modulo 6, the link attached to joint 1 and 2 has index 1.

We start by restricting the possible candidates of bond diagrams of such 6R linkages. 
We may distinguish ``long bond connections'' connecting joints 1-4, 2-5, or 3-6, and ``short bond connections''
connecting $i$-$i+2$ for $i=1,\dots 6$ modulo $6$. By Lemma~\ref{lem:bcon1}, consecutive joints are not connected.

\begin{Theorem} \label{thm:bdia}
For any mobile 6R linkage, the following conditions on the bonds are known.
\begin{enumerate}
\item Any joint is attached to at most 4 conjugate pairs of bonds (counted with multiplicities).
\item If $\dim(L_{i,i+1,i+2})=\dim(L_{i+1,i+2,i+3})=6$, then there is no bond connecting joints $i$ and $i+3$.
\item If $\dim(L_{i,i+1,i+2})=\dim(L_{i+3,i+4,i+5})=6$, and the algebraic degree of $C_{i+2,i+5}$ is bigger
	than 4, then the linkage is a Dietmaier linkage (see Example~\ref{ex:dietb}).
\item If $\dim(L_{i,i+1,i+2})=6$ and $\dim(L_{i+3,i+4,i+5})=8$, and the algebraic degree of $C_{i+2,i+5}$ is bigger
	than 6, then it is 8.
\end{enumerate}
\end{Theorem}

\begin{proof}
(1): by Theorem~\ref{thm:adeg}, the number $k$ of pairs of bonds attached to joint $i$ is equal to the algebraic
degree of $C_{i-1,i}$. Since $C_{i-1,i}$ is a line parametrizing rotations around a fixed axis, $k$ is the
number of configurations with fixed joint parameter $t_i$. Since the maximal number of configurations of a 
non-mobile 5R linkage with no three consecutive axes incident to a single point is 4, the claim follows.

For (2), we refer to \cite[Lemma~5.6]{li}. For (3) and (4), we refer to \cite{hlss}.
\end{proof}

When the dimensions of all coupling spaces $L_{i,i+1,i+2}$ is 8, then there are no short connections. Let us assume
that the number of bonds connecting joints 1 and 4 is $k_{14}$, the number of bonds connecting joints 2 and 5 is $k_{25}$,
and the number of bonds connecting joints 3 and 6 is $k_{36}$. Without loss of generality, we may assume that
$k_{14}\le k_{25}\le k_{36}$; also, we have $k_{15}\ge 1$ by Proposition~\ref{thm:frozen} and $k_{36}\le 4$ by Theorem~\ref{thm:bdia}.
Here is a summary about what is known for these cases.

\begin{itemize}
\item All linkages with $k_{14}=k_{25}=1$ are known: there is one family \cite{hss1} with $k_{36}=1$ and another family \cite{sharp}
	with $k_{36}=2$.
	Both families are maximal, i.e. they are irreducible components of the variety ${\cal M}_6$ of movable 6R loops.
\item All linkages with $k_{25}=k_{36}=4$ are known: there is one family with $k_{14}=2$ \cite{ls} and two families
	with $k_{14}=4$ \cite{hlss}. All three families are maximal.
\item The family of line symmetric linkages has $k_{14}=k_{25}=k_{36}=2$. It is maximal. Another family with the same bond
	diagram can be found in \cite{par}. But we do not know if this family is maximal, or if there
	are other families with $k_{14}=k_{25}=k_{36}=2$.
\item There are examples with linkages with $k_{14}=k_{25}=k_{36}=3$ with reducible configuration space (see below). But we do not
	know if there are other examples.
\item For any triple $(k_{14},k_{25},k_{36})$ not covered by the above cases, we do not know if there are any linkages.
\end{itemize}

\begin{Example} \label{ex:333}
Let $h_1,h_2,h_3$ be three dual quaternions corresponding to three random lines. Set 
\[ h_4:=h_1,\ h_5:=h_2,\ h_6:=h_3 . \]
The configuration space of the closed linkage specified by the axes corresponding to $h_1,\dots,h_6$ is a reducible curve with 
four components: three lines parametrizing rotations around a coincident axis, and one curve parametrizing the motion of a line
symmetric linkage. The first three lines intersect in a common point, corresponding to the initial configuration with three
coincident lines. The fourth component does not meet the other three.
\end{Example}

The question whether a known family of linkages is maximal is more difficult when short connections are present. In this case,
maximality is known only for just one family, namely Dietmaier's linkage \cite{dietmaier}. The situation is more complicated because
one has to exclude that the family under consideration is a specialization of another family yet unknown for which the dimension 
of the coupling spaces is 8. In the case of Dietmaier's linkage, it is possible to prove maximality by 
semicontinuity of the genus of the configuration curve (see \cite{hlss}).

Let us consider only bond diagrams with at least one short connection. The thesis \cite{li} contains
a complete list of linkages with three or four conjugated pair of bonds (i.e. the bond diagram contains three or four lines; some
of these diagrams are depicted in Figure~\ref{fig:compb}). For most diagrams compatible with Theorem~\ref{thm:bdia} and 
Lemmas~\ref{lem:bcon1}, \ref{lem:bcon2}, we do not know if they are bond
diagrams of movable linkages, and for many others we have examples but we do not know if the known examples are all linkages 
with the bond diagram under consideration. There are just two diagrams with more than four lines and at least one short connection
for which all linkages are known, namely the diagram of Wohlhart's partially symmetric linkage~\cite{partsym} and its third isomerization 
shown in Figure~\ref{fig:bwbw}. Isomerization is a technique introduced in \cite{iso}
that allows to construct new families from known ones by interchanging two links; it is possible only if their affected joints
satisfy the Bennett condition.

\begin{figure}
\begin{center}
\includegraphics{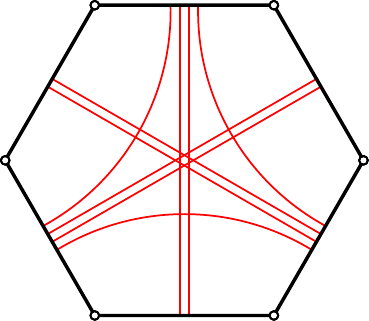}\quad\quad
\includegraphics{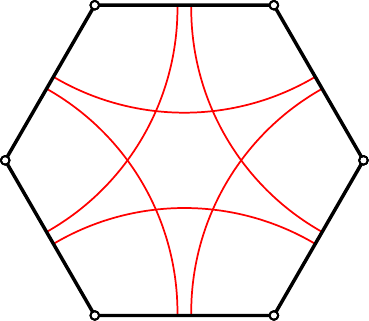}
\end{center}
\caption{The bond diagrams of the Wohlhart's partially symmetric linkage and its third isomerization.
	For a generic partially symmetric linkage, three consecutive triples of axes satisfy the Bennett condition and three do not,
	which one can be seen from the short connections in the diagram. Isomerization works only if all six triples satisfy
	the Bennett condition.}
\label{fig:bwbw}
\end{figure}

%
%

\bibliographystyle{plain}
\bibliography{mathrob}

\end{document}